\documentclass[10pt]{llncs}
\usepackage{amssymb,amsmath,latexsym}
\usepackage{amsfonts}
\usepackage{amsfonts}
\usepackage{graphicx}
\usepackage{epstopdf}
\usepackage{amssymb}
\usepackage{centernot}
\usepackage{chemfig}
\usepackage{tikz}
\usetikzlibrary{arrows}

\vspace{5cm}

\begin{document}
	
	\label{'ubf'}  
\setcounter{page}{1}                                 

\markboth {\hspace*{-9mm} \centerline{\footnotesize \sc
            Label Languages of 8-directional Array P-
System  }
                 }
                { \centerline                           {\footnotesize \sc  
            Kalpana Mahalingam, Raghavan Rama and Williams Sureshkumar } \hspace*{-9mm}              
               }

\vspace*{-2cm}
\begin{flushleft}
    {
    }
           \vspace*{0.3cm}       
\end{flushleft}

\begin{center}
{ 
       {\Large \textbf { \sc Label Languages of 8-directional Array P-
System
                               }
       }
\\

\medskip

{\sc Williams Sureshkumar$^{1}$, Kalpana Mahalingam$^{2}$, Raghavan Rama$^{2}$}\\ 

{\footnotesize $^{1}$Department of Computer Science and Engineering, }\\
{\footnotesize Saveetha University,  Chennai-602105. India
}\\
{\footnotesize $^{2}$Department of Mathematics, }\\
{\footnotesize Indian Institute of Technology, Madras,  Chennai-600036. India
}\\
{\footnotesize e-mail: {\it wisureshkumariit@gmail.com, kmahalingam@iitm.ac.in, ramar@iitm.ac.in}}
}
\end{center}

\thispagestyle{empty}

\hrulefill

\begin{abstract}
An 8-directional array P system is one where the rewriting of an array can happen in any 8-directions. The array rules of such a system are labelled thus resulting in a labelled 8-directional array P system. The labelling is not unique and the label language is obtained by recording the strings over the labels used in any terminating  derivation of the P system. The system is shown to generate interesting pictures. The label language is compared with Chomsky hierarchy. 

\end{abstract}

\section{Introduction}
A P system is a new computing model abstracting the biological happening in membranes. Hence the system is also called \textit{membrane system}. This computing model has several variants \cite{Gh.Paun-Rozenberg-Salomaa-2010}. Most P system variants are computationally universal exhibiting the power of the systems. Recently, there were several research papers on array P systems. The arrays considered in \cite{Rama-krishna-Kamala-2001} are set to evolve in a P system to generate various pictures. In \cite{Ceterchi-Madhu-Paun-Subramanian-2003} the authors set a collection of pictures made up of symbols to evolve using array rules mostly of isotonic type. The key idea was the construction of array language of halting P system. Hence it is interesting to look for the nature of P systems which have well defined halting configurations. Ginsburg and Spanier introduced controlled grammars, where the rules were applied in a restricted manner. The rules in any context-free grammar $G$ were labelled uniquely by distinct symbols of a set $C$ and the Context-free Grammar $($CFG$)$ $G$ is regarded as a mechanism for translating strings over $C$ into a language $L(G, C)$ of strings. Every   
member of $L(G, C)$ is generated by a sequential application of rules of $G$ and the sequence being labelled by words over $C$. Associating a 'label' string with a halting computation of a P system is defined in \cite{Kamala-Paun-Ajeesh-Kamala-2013}. The authors in this paper(\cite{Kamala-Paun-Ajeesh-Kamala-2013}) realised the difficulty in associating labels with the rules in a parallel system and have overcome the difficulty by assigning the same label to more than rule. Hence they named the strings over the labels as control strings rather than label strings.

\par In \cite{Kamala-Paun-Ajeesh-Kamala-2013} the authors looked at the control string associated with a computation of a P system with multi set of objects. The study of control languages of Tissue P-systems is initiated in \cite{pan-tissue}. In \cite{Suresh-Rama-2015} the authors looked at the regulating evolution of an isotonic array P system where the evolution rules were either regular isotonic or context-free isotonic as defined in \cite{Rosenfeld-1971}. The authors in this paper(\cite{Suresh-Rama-2015}) also introduce a new type of isotonic rule called \textit{restricted monotonic} type which is different from the array rules used in \cite{Ceterchi-Madhu-Paun-Subramanian-2003}. Some of the interesting P systems which use arrays as data structures can be seen in \cite{Freund-2000,Kamala-Anindya-1989}.

\par In \cite{suresh-8dp} the authors defined new array systems  called \textit{8-directional array grammar} and \textit{8-directional array P-System}. The interpretation and manipulation of the data structure `string' will be like `turtle-like' graphs with possibilities to turn in multiples of 45 degrees. This array P system naturally can be seen to produce several interesting arrays of both rectangular and non rectangular in nature. Hence,  the authors in \cite{suresh-8dp} compare  8-directional array P System with the power of existing  array models.

 In this paper, we are interested in looking at the evolution of labelled array P system. In such a system , every array rule will be labelled. The labelling need not be unique. A label string is one in which the label symbols of the rules applied are concatenated in sequence and the associated derivation should be an halting one. Collection of such label strings form a label language.

\par In section 2, we give two definitions. They are 8-directional Array Grammar and 8-directional Array P System ($8dAPS$). In section 3, we define Labelled 8-directional Array P System ($L8dAPS$). A label string over the labels of  the array P systems always meets a halting configuration. The set of such label string makes a label language. The set of arrays in the halting configuration may be some interesting pictures which other array systems could not generate. In section 4, we compare the label language with Chomsky hierarchy. We provide some comments and future direction research in section 5.

\section{8-Directional Array P Systems}

An array grammar called  8-directional array grammar is defined in this section. An array P system that works with arrays as data structure and 8-directional array rules as control structure is called a 8-directional array P system. It is denoted by $8dAPS$.

\begin{definition}\cite{Suresh-Rama-cmc-2015} An 8-directional array grammar is defined to be a quadruple $G = (N,~ T,~ P,~ S)$, where
\begin{enumerate}
\item $N$ is a finite non-empty set of symbols called non-terminals.

\item $T$ is a finite non-empty set of symbols called terminals such that it is disjoint from $N$ i.e. $N \cap T$ = $\emptyset$.

\item $P$ is a finite non-empty set of $\theta$-rotation rules of the form
\begin{center}
$A\rightarrow \beta^{\theta}$
\end{center}
\begin{center}
or
\end{center}

\begin{center}
~~~~~~~~~~~~~~~~~~$\alpha^{\theta}\rightarrow \beta^{\theta},~~ 2 \leq \left|\alpha\right|\leq \left|\beta\right|$,
\end{center}

\noindent where $A \in N$, $\alpha, \beta \in (N \cup T)^{+}$, $\alpha$ contains exactly one non-terminal symbol and all other symbols in $\alpha$ are terminals, $\theta\in\Big{\{}0, \frac{\pi}{4}, \frac{\pi}{2}, \frac{3\pi}{4}, \pi, \frac{5\pi}{4}, \frac{3\pi}{2}, \frac{7\pi}{4} \Big{\}}$. While applying the former type of rule, $A$ is rewritten by $\alpha$ in the direction of $\theta$ such that the leftmost symbol of $\alpha$ is placed in the position of $A$. For the later type of rule, $\alpha$ is replaced by $\beta$ in the direction of $\theta$ such that the first symbol of $\beta$ is placed in the position of the first symbol of $\alpha$.

\item $S \in N$ is the start symbol.
\end{enumerate}
\end{definition}

\begin{remark} 
For the rules of the form $A\rightarrow \beta^{\theta}$, the symbols following $A$ are to be shifted by $\left|\beta - 1\right|$ positions in the direction of $\theta$. For the rules of the form $\alpha^{\theta}\rightarrow \beta^{\theta}$ the symbols following $\alpha$ are to be shifted by $\left|\beta \right| - \left|\alpha \right|$ positions in the direction of $\theta$.
\end{remark}

\begin{example} The rule $\begin{array}{c}
X
\end{array}  \rightarrow \begin{array}{c}
(aXY)^{\frac{3\pi}{4}}
\end{array}$ means, while applying the rule to any string of the form $\alpha X \beta$, the resultant array will be
\begin{center}
$\begin{array}{cccc}
 Y  &         &   &          \\
    &   X     &   &           \\
    &\alpha~  & a & ~\beta 
\end{array}$
\end{center}
 \noindent To the above array, if we apply the rule $(aX)^{\frac{3\pi}{4}} \rightarrow (bcdZ)^{\frac{3\pi}{4}}$, the resultant array will be
\begin{center}
$\begin{array}{cccccc}
     Y    &   &        &           &   &   \\
          & Z &        &           &   &    \\
          &   &   d    &           &   &    \\
          &   &        &    c      &   &    \\
          & b &        & \alpha~   & b & ~\beta
\end{array}$ 
\end{center}
\end{example}

\begin{definition}\cite{Suresh-Rama-cmc-2015}\label{def4} An 8-directional Array P System ($8dAPS$) of degree $m (\geq1)$ is a construct 
\begin{center}
$\Pi = (V,~ T,~ \mu,~I_1, \dots , I_m,~ (R_1, \rho_1), \dots , (R_m, \rho_m),~ i_o)$
\end{center}
where $V$ is the total
alphabet, $T \subseteq V$ is the terminal alphabet, $\mu$ is a membrane structure
with $m$ membranes labelled in a one-to-one manner with $1, 2, \dots ,m; I_1, \dots , I_m$ are finite sets of initial arrays over $V$ associated with the $m$ regions of $\mu;~ R_1, \dots , R_m$ are
finite sets of $\theta$-rotation rules over $V$ associated
with the $m$ regions of $\mu$; $\rho_1, \dots , \rho_m$ are partial order relations over $R_1, \dots , R_m$. The rules in $R_i$ are of the form $A\rightarrow \alpha^{\theta}~ (tar)$, or $\alpha^{\theta}\rightarrow \beta^{\theta}~ (tar)$, $2 \leq \left|\alpha\right|\leq \left|\beta\right|$, where $tar$ indicates the target location of the output array obtained by applying such rules. The $tar$ can be $here$, $out$ or $in$. Here $A \in (V\setminus T)$, $\alpha$ contains exactly one non-terminal symbol and all other symbols in $\alpha$ are terminals, $\beta \in V^{+}$ and $\theta\in \Big{\{}0, \frac{\pi}{4}, \frac{\pi}{2}, \frac{3\pi}{4}, \pi, \frac{5\pi}{4}, \frac{3\pi}{2}, \frac{7\pi}{4} \Big{\}}$. There can be more than one rule with $A$ or $\alpha$ on its left hand side. The array produced by using this rule will go to the membrane indicated by $tar$; $i_o$ is the output membrane.
\end{definition}

\par We start from an initial configuration of the system and proceed iteratively, by transition steps performed by using the $\theta$-rotation rules in parallel, to all arrays that can be rewritten, obeying the priority relations, and collecting the terminal arrays thus generated in a designated output membrane.

\par Note that each array is processed by one rule only, the parallelism refers here to processing simultaneously all available arrays by all applicable $\theta$-rotation rules. If several rules can be applied to an array, may be in several places, then we identify only one possible location to apply a possible rule. It is important to have in mind the fact that the evolution of the arrays is not independent of each other, but interrelated in two ways : (1) if we have priorities, a rule $r_1$ applicable to an array $\mathcal{A}$ can forbid the use of another rule, $r_2$, for rewriting another array, $\mathcal{B}$, which is present at that time in the same membrane. In the next step if $r_1$ is not applicable to $\mathcal{B}$ or to the array $\mathcal{A^{'}}$ obtained from $\mathcal{A}$ by using $r_1$, then it is possible that the rule $r_2$ can now be applied to $\mathcal{B}$; (2) even without priorities, if an array  $\mathcal{A}$ can be rewritten for ever, in the same membrane or on an itinerary through several membranes, and if this cannot be avoided, then all arrays are lost, because the computation never stops. The arrays collected in the output membrane are then lost.
\par A computation is successful only if it halts, a configuration is reached where no rule can be
applied to the existing arrays. The result of a halting computation consists of the arrays composed only of symbols from T (terminal symbols) placed in the membrane with label $i_o$.

\section{Labelled $8dAPS$}
In this section we introduce labelled $8dAPS$. We illustrate the model with a few interesting examples which halt always on digitized pictures. An exactly labelled $8dAPS$ is an 8-directional array P system where every rule is labelled either with a symbol or '$\lambda$'. The assignment of labels to rules need not be unique. The array rules are applied in a parallel distributed manner. The arrays are derived by any computation of a P system and such derivations of the system are labelled by strings over the labels. As the system being exactly labelled, the system is defined with a set of output arrays to which a derivation is built. The labels of the rules  applied are concatenated in sequence with the application of the rules. Such a collection of strings will be called as label language.

\begin{definition} A Labelled 8-directional Array P System ($L8dAPS$) $\Pi$ of degree $m (\geq1)$ is a construct
 
\begin{center}
$\Pi = (V, T, \mu, I_1, \dots , I_m, (R_1, \rho_1), \dots , (R_m, \rho_m), i_o, lab)$
\end{center}
where  $V, T, \mu, I_1, \dots , I_m, R_1, \dots ,R_m, \rho_1, \dots , \rho_m$, $i_o$ are same as in Definition \ref{def4}, $lab$ is a finite set of alphabet, which is used for labelling the rules. Let  $R = \bigcup_{i=1}^{m}R_i$. Here we assign a label to
every rule in $R$ where the labels are chosen from the finite
alphabet $lab$ or the labels can be $\lambda$ (empty label). Define a
function $f : R \rightarrow lab \cup \{\lambda\}$ called
a labelling  function that assign a label to each rule in $R$. Noting that more than one rule may have the same label, but the same rule in different membranes cannot be assigned different labels. We
extend the labelling for a label sequence $S = r_1~ r_2~ \dots~
r_k \in R^{*}$ as follows : $f(r_1~ r_2~ \dots~ r_k) = f(r_1) f(r_2 \dots r_k)$, where $r_i$ represents a rule in $R$. A transition
$C \stackrel{b}\Rightarrow C^{'}$ between two successive
configurations uses only rules with the same label $b$ and rules
labelled with $\lambda$. If at least one rule has a label $b\in
lab$ then the transition is called $\lambda$-restricted
transition. If we allow all rules with $\lambda$ label then the transition is called
$\lambda$-unrestricted transition (or $\lambda$-transition).
\end{definition}

\par A label string of input symbols (over $lab$) is said to be
generated if all its symbols are consumed while $\Pi$ transits from an initial configuration to a halting configuration. The output arrays in such a halting configuration are shown in $F$. The set of all label strings generated in this way by computations in a $L8dAPS$ $\Pi$ is
denoted by $L_{\lambda}8dAP(\Pi)$. The subscript indicates the
fact that $\lambda$-steps (all rules applied in one step can have
$\lambda$ label) are permitted. When only steps where at least one
rule with a non-empty label is used, the generating language is
denoted by $L8dAP(\Pi)$. The family of languages $L8dAP(\Pi)$
associated with $L8dAPS$ with at most $m$ membranes is denoted by $L8dAP_m$. In the
unrestricted case, the corresponding language family is denoted by
$L_{\lambda}8dAP_m$. If the number of membranes is unbounded,
then the subscript $m$ is replaced with $\star$.

\par We now illustrate the system with some interesting examples. In example \ref{ex1} the labelling language is regular where as in example \ref{ex2} the label language is context-free. From these examples one can observe that the halting configuration set can contain both rectangular and non rectangular arrays.

\begin{example}\label{ex1} Consider the array language $L_{star}$ consisting of the star shaped arrays over $\big{\{}x \big{\}}$ with each of the star having eight arms of equal length, (i.e.,)  

$L_{star} = \Bigg\{ \begin{array}{ccccc}
                                 x~     &         &  x~  &         & x~  \\
                                        & x~      &  x~  &   x~    &    \\
                                 x~     & x~      &  x~  &   x~    & x~  \\
                                        & x~      &  x~  &   x~    &    \\
                                 x~     &         &  x~  &         & x~

 \end{array}~~,~~\begin{array}{ccccccc}
 
                        x~      &         &         &  x~    &         &       &  x~   \\
                                &  x~     &         &  x~    &         &  x~   &         \\
                                &         & x~      &  x~    & x~      &       &         \\
                        x~      &  x~     & x~      &  x~    & x~      &  x~   &  x~      \\
                                &         & x~      &  x~    & x~      &       &         \\
                                &  x~     &         &  x~    &         &  x~   &         \\
                        x~      &         &         &  x~    &         &       &  x~

 \end{array},\dots \Bigg\}$. \\
 
\noindent The $LIAPS$ $\Pi_1$ with $three$ membranes is given as, 


\noindent Let $\Pi_1 = \Big(\Big{\{}A, B, C, D, E, F, G, H, A_1, B_1, C_1, D_1, E_1, F_1, G_1, H_1, x \Big{\}},~ \Big{\{}x \Big{ \}},~[_1[_2]_2$\\$~~~~~~~~~~~~~~~~[_3]_3]_1, ~I_1, I_2,~ I_3,~ R_1,~ R_2,~ R_3,~ 3, \Big{\{}a \Big{ \}}\Big)$, \\ where  $I_1 =\Bigg\{ \begin{array}{ccc}
                                            
                                        D      &  C  &   B  \\
                                        E      &  a  &   A  \\
                                        F      &  G  &   H

 \end{array} \Bigg\}$, $I_2 = I_3 = \emptyset$ and $F = \Bigg{\{}(\phi$, $\phi$, $\mathcal{A})$ $\big{|} \mathcal{A}\in L_{star} \Bigg{\}}$.

\noindent The set of $\theta$-rotation rules are given by \\

\noindent $R_1 = \Bigg\{~\Bigg\{
~1)~a:\begin{array}{c}
A
\end{array} \rightarrow \begin{array}{c}
(xA_1)^{0}
\end{array},~~2)~a:\begin{array}{c}
B
\end{array} \rightarrow \begin{array}{c}
(xB_1)^{\frac{\pi}{4}}
\end{array},~~3)~a:\begin{array}{c}
C
\end{array} \rightarrow \begin{array}{c}
(xC_1)^{\frac{\pi}{2}}
\end{array},\\~~~~~~~~~~~~~~4)~a:\begin{array}{c}
D
\end{array} \rightarrow \begin{array}{c}
(xD_1)^{\frac{3\pi}{4}}
\end{array},~~5)~a:\begin{array}{c}
E
\end{array} \rightarrow \begin{array}{c}
(xE_1)^{\pi}
\end{array},~~6)~a:\begin{array}{c}
F
\end{array} \rightarrow \begin{array}{c}
(xF_1)^{\frac{5\pi}{4}}
\end{array},\\~~~~~~~~~~~~~~7)~a:\begin{array}{c}
G
\end{array} \rightarrow \begin{array}{c}
(xG_1)^{\frac{3\pi}{2}}
\end{array}~~\Bigg\}~>~\Bigg\{~8)~a:\begin{array}{c}
H
\end{array} \rightarrow \begin{array}{c}
(xH_1)^{\frac{7\pi}{4}}
\end{array},in_2~,\\~~~~~~~~~~~~~~9)~a:\begin{array}{c}
H
\end{array} \rightarrow \begin{array}{c}
(xx)^{\frac{7\pi}{4}}
\end{array},in_3~\Bigg\}\Bigg\}$

\noindent $R_2 = \Bigg\{~\Bigg\{
10)~a:\begin{array}{c}
A_1
\end{array} \rightarrow \begin{array}{c}
(A)^{0}
\end{array},~~~~11)~a:\begin{array}{c}
B_1
\end{array} \rightarrow \begin{array}{c}
(B)^{\frac{\pi}{4}}
\end{array},~~12)~a:\begin{array}{c}
C_1
\end{array} \rightarrow \begin{array}{c}
(C)^{\frac{\pi}{2}}
\end{array},\\~~~~~~~~~~~~~~13)~a:\begin{array}{c}
D_1
\end{array} \rightarrow \begin{array}{c}
(D)^{\frac{3\pi}{4}}
\end{array},~~14)~a:\begin{array}{c}
E_1
\end{array} \rightarrow \begin{array}{c}
(E)^{\pi}
\end{array},~~15)~a:\begin{array}{c}
F_1
\end{array} \rightarrow \begin{array}{c}
(F)^{\frac{5\pi}{4}}
\end{array},\\~~~~~~~~~~~~~~16)~a:\begin{array}{c}
G_1
\end{array} \rightarrow \begin{array}{c}
(G)^{\frac{3\pi}{2}}
\end{array}~~\Bigg\}~>~17)~a:\begin{array}{c}
H_1
\end{array} \rightarrow \begin{array}{c}
(H)^{\frac{7\pi}{4}}
\end{array},out~\Bigg\}$

\noindent $R_3 = \Bigg\{
18)~a:\begin{array}{c}
A_1
\end{array} \rightarrow \begin{array}{c}
(x)^{0}
\end{array},~~~~19)~a:\begin{array}{c}
B_1
\end{array} \rightarrow \begin{array}{c}
(x)^{\frac{\pi}{4}}
\end{array},~~20)~a:\begin{array}{c}
C_1
\end{array} \rightarrow \begin{array}{c}
(x)^{\frac{\pi}{2}}
\end{array},\\~~~~~~~~~~~21)~a:\begin{array}{c}
D_1
\end{array} \rightarrow \begin{array}{c}
(x)^{\frac{3\pi}{4}}
\end{array},~~22)~a:\begin{array}{c}
E_1
\end{array} \rightarrow \begin{array}{c}
(x)^{\pi}
\end{array},~~23)~a:\begin{array}{c}
F_1
\end{array} \rightarrow \begin{array}{c}
(x)^{\frac{5\pi}{4}}
\end{array},\\~~~~~~~~~~~24)~a:\begin{array}{c}
G_1
\end{array} \rightarrow \begin{array}{c}
(x)^{\frac{3\pi}{2}}
\end{array}~~\Bigg\}$. \\

\noindent $L8dAPS$ $\Pi_1$ generating an exact label language $L8dAP(\Pi_1) = \Big{\{}a^{16n}~ \big{|}~ n \geq 1 \Big{\}}$, while it halts on an array that belongs to the set of picture configurations in $F$. The working of $L8dAPS$ $\Pi_1$ with $three$ membrane is as follows:\\

\par The initial array in membrane $one$ contains the array  

\begin{center}
$\begin{array}{ccc}
                                            
                                        D      &  C  &   B  \\
                                        E      &  x  &   A  \\
                                        F      &  G  &   H

 \end{array}$.
\end{center} 
  Applying rules 1, 2, 3, 4, 5, 6, 7  in membrane 1 and applying one of the lower priority rules say rule 8,the resulting array moves to region 2. In region 2, the variables $A_1$, $B_1$, $C_1$, $D_1$, $E_1$, $F_1$, $G_1$ are renamed as $A$, $B$, $C$, $D$, $E$, $F$, $G$ by applying the rules 10, 11, 12, 13, 14, 15, 16 and finally the lower priority rule 17 is applied to rename the variable $H_1$ by $H$. The resultant array is sent back to membrane 1. The process can be repeated to generate the array of eight arms of equal length  over ${\{}x{\}}$ in membrane 3. To halt the computation  rule 9 is applied instead of rule 8 in membrane 1 and the resultant array is  sent to membrane 3. In  membrane 3,  rules 18 to 24 are applied.
  
\end{example}

\begin{example}\label{ex2}
In this example, $L8dAPS$ $\Pi_2$ halts on a $swastik$ pattern $L_{s}$ of each arm length $\geq$ 3 in membrane 5, and it generates the exact label language which is context-free in nature. The $L8dAPS$ with 5 membranes is given as,

\noindent $\Pi_2 = \Big(\Big{\{}A, B, C, D, E, A_1, B_1, C_1, D_1, E_1, U, X, Y, Z, 0 \Big{\}},~ \Big{\{}0 \Big{ \}},~[_1[_2]_2[_3[_4]_4[_5]_5]_3]_1,\\~~~~~~~~~~~~I_1, I_2,~ I_3, I_4, I_5,~ R_1,~ R_2,~ R_3,~R_4,~R_5, 5, \Big{\{}a, b \Big{ \}}\Big)$,  where \\ $I_1 = \Bigg\{ \begin{array}{ccc}
                                            
                                               &  A  &     \\
                                          D    &  0  &   B  \\
                                               &  C  &

 \end{array} \Bigg\}$, $I_2 = I_3 = I_4 = I_5 = \emptyset$ and $F = \Bigg{\{}(\phi$, $\phi$, $\phi$, $\phi$, $\mathcal{B})$ $\big{|} \mathcal{B} \in L_{s} \Bigg{\}}$.

\noindent The set of $\theta$-rotation rules are given by \\

\noindent $R_1 = \Bigg\{~\Bigg\{
~1)~a:\begin{array}{c}
A
\end{array} \rightarrow \begin{array}{c}
(XA_1)^{\frac{\pi}{2}}
\end{array},~~2)~a:\begin{array}{c}
B
\end{array} \rightarrow \begin{array}{c}
(YB_1)^{0}
\end{array},~~3)~a:\begin{array}{c}
C
\end{array} \rightarrow \begin{array}{c}
(ZC_1)^{\frac{3\pi}{2}}
\end{array}~\Bigg\}~\\~~~~~~>~\Bigg\{~4)~a:\begin{array}{c}
D
\end{array} \rightarrow \begin{array}{c}
(UD_1)^{\pi}
\end{array},in_2~,~5)~a:\begin{array}{c}
D
\end{array} \rightarrow \begin{array}{c}
(UE_1)^{\pi}
\end{array},in_3~\Bigg\}\Bigg\}$

\noindent $R_2 = \Bigg\{~\Bigg\{
~6)~a:\begin{array}{c}
A_1
\end{array} \rightarrow \begin{array}{c}
(A)^{0}
\end{array},~~7)~a:\begin{array}{c}
B_1
\end{array} \rightarrow \begin{array}{c}
(B)^{0}
\end{array},~~8)~a:\begin{array}{c}
C_1
\end{array} \rightarrow \begin{array}{c}
(C)^{0}
\end{array}~\Bigg\}~>~\\~~~~~~~~~~~~~~9)~a:\begin{array}{c}
D_1
\end{array} \rightarrow \begin{array}{c}
(D)^{0}
\end{array},out~\Bigg\}$

\noindent $R_3 = \Bigg\{
~10)~b:\begin{array}{c}
A_1
\end{array} \rightarrow \begin{array}{c}
(0A_1)^{\frac{\pi}{2}}
\end{array},in_4~~>~~11)~b:\begin{array}{c}
B_1
\end{array} \rightarrow \begin{array}{c}
(0B_1)^{0}
\end{array},in_4~~>~~\\~~~~~~~~~~~12)~b:\begin{array}{c}
C_1
\end{array} \rightarrow \begin{array}{c}
(0C_1)^{\frac{3\pi}{2}}
\end{array},in_4~~>~~13)~b:\begin{array}{c}
E_1
\end{array} \rightarrow \begin{array}{c}
(0E_1)^{\pi}
\end{array},in_4~\Bigg\}$

\noindent $R_4 = \Bigg\{~\Bigg\{
14)~b:\begin{array}{c}
X
\end{array} \rightarrow \begin{array}{c}
(0)^{0}
\end{array},out~~>~~15)~b:\begin{array}{c}
A_1
\end{array} \rightarrow \begin{array}{c}
(A)^{0}
\end{array},out~~>~~\\~~~~~~~~~~~~~16)~b:\begin{array}{c}
Y
\end{array} \rightarrow \begin{array}{c}
(0)^{0}
\end{array},out~~>~~17)~b:\begin{array}{c}
B_1
\end{array} \rightarrow \begin{array}{c}
(B)^{0}
\end{array},out~~>~~$

$\\~~~~~~~~~~~~~18)~b:\begin{array}{c}
Z
\end{array} \rightarrow \begin{array}{c}
(0)^{0}
\end{array},out~~>~~19)~b:\begin{array}{c}
C_1
\end{array} \rightarrow \begin{array}{c}
(C)^{0}
\end{array},out~~>~~$

$\\~~~~~~~~~~~~~20)~b:\begin{array}{c}
U
\end{array} \rightarrow \begin{array}{c}
(0)^{0}
\end{array},out~~>~~21)~b:\begin{array}{c}
E_1
\end{array} \rightarrow \begin{array}{c}
(E)^{0}
\end{array},in_5\Bigg\}$

\noindent $R_5 = \Bigg\{
22)~b:\begin{array}{c}
A
\end{array} \rightarrow \begin{array}{c}
(0)^{0}
\end{array},~23)~b:\begin{array}{c}
B
\end{array} \rightarrow \begin{array}{c}
(0)^{0}
\end{array},~24)~b:\begin{array}{c}
C
\end{array} \rightarrow \begin{array}{c}
(0)^{0}
\end{array}~,~\\~~~~~~~~~~25)~b:\begin{array}{c}
E
\end{array} \rightarrow \begin{array}{c}
(0)^{0}
\end{array}\Bigg\}$.

%
%
%
%
%
 
\noindent The label language generated by $\Pi_2$ is $L8dAP(\Pi_2) = \Big{\{}a^{8n+4}b^{8n+20} ~~\big{|}~~ n \geq 0 \Big{\}}$. For each label string $a^4b^{20}, a^{12}b^{28}\dots\dots$,  $\Pi_2$ generates the $swastik$ pattern of each arm length $3, 4, \dots$ respectively whihc are in $L_{s}$.

$L_{s} = \Bigg\{ 
\begin{array}{ccccc}
                                    0   &        &  0  &   0 & 0  \\
                                    0   &        &  0  &     &    \\
                                    ~0~   &  ~0~     &  ~0~  &   ~0~ & 0  \\
                                        &        &  0  &     & 0   \\
                                    0   &  0     &  0  &     & 0

 \end{array}~,~\begin{array}{ccccccc}
                              0 &       &        &  0  &  0  & 0  & 0       \\
                              0 &       &        &  0  &     &    &        \\
                              0 &       &        &  0  &     &    &        \\
                              ~0~ &   0~   &  0~     &  ~0~  &   0~ & 0~  & 0     \\
                                &       &        &  0  &     &    & 0       \\
                                &       &        &  0  &     &    & 0       \\
                              0 &   0   &  0     &  0  &     &    & 0

 \end{array}~,~\dots\dots \Bigg\}$

\noindent The working of $L8APS$ $\Pi_2$ with five membranes is as follows: Starting with an initial array
\begin{center}

$\begin{array}{ccc}
                                            
                                               &  A  &     \\
                         \mathcal{Z} =              D    &  0  &   B  \\
                                               &  C  &

\end{array}$
\end{center}

\noindent in membrane 1. To generate a $swastik$ pattern of each arm length $n+3$, we first apply rules 1, 2, 3 and 4 in membrane 1 to $\mathcal{Z}$. The resulting array moves to membrane 2. Rules 6, 7, 8 and 9 are applicable now and the array moves to membrane 1. The process can be repeated $n$ times and the resulting array is as follows:
\begin{center}
$\begin{array}{ccccccccccccc}
    &    &     &  ~  &         &         &  A~    &         &       &  ~   &    &    &    \\
    &    &     &  ~  &         &         &  .~    &         &       &  ~   &    &    &    \\
    &    &     &  ~  &         &         &  .~    &         &       &  ~   &    &    &    \\
    &    &     &  ~  &         &         &  .~    &         &       &  ~   &    &    &     \\
    &    &     &     &  ~      &         &  X~    &         &  ~    &      &    &    &      \\
    &    &     &     &         & ~       &  X~    & ~       &       &      &    &    &      \\           
  D &  .~ &  .~  &  .~  &  U~     & U~      &  0~    & Y~      &  Y~   &  ~.   &  ~. &  ~. &  B   \\
    &    &     &     &         & ~       &  Z~    & ~       &       &      &    &    &     \\
    &    &     &     &  ~      &         &  Z~    &         &  ~    &      &    &    &      \\
    &    &     &     &         &         &  .~    &         &       &  ~   &    &    &      \\
    &    &     &    ~&         &         &  .~    &         &       &  ~   &    &    &      \\
    &    &     &    ~&         &         &  .~    &         &       &  ~   &    &    &      \\
    &    &     &   ~ &         &         &  C~    &         &       &  ~   &    &    &   
                                  
 \end{array}$
\end{center}

\noindent  Now, apply rules 1,2, 3 and 5 once in membrane 1 and expel the array to membrane 3, the label string so far will be $a^{8n+4}$ . Upper arm of a resulting array consists of a $0$ followed by $(n+1)$ $X$'s followed by $A_1$. Other arms consists of $n+1$ copies $Y$'s, $Z$'s and $U$'s in place of $X$'s followed $B_1$, $C_1$ and $E_1$ respectively.

\par To replace each $(n+1)$ $X$'s by $0$'s, apply rule 10 in region 3 and rule 14 in region 4 alternatively for $n+1$ times. After replacing all $X$'s, apply rule 10 in region 3 and rule 15 in region 4 once. The corresponding label string will be $a^{8n+4}b^{2(n+1)+2}$. Similarly, using rules 11, 12 and 13 in region 3 and rules 16, 17, 18, 19, 20 and 21 in region 4 all $Y$'s, $Z$'s and $U$'s in the array can be replaced by $0$'s. The array is expelled to membrane 5. The corresponding label string is $a^{8n+4}b^{8(n+1)+8}$. In membrane 5, apply the only possible rules 22, 23, 24 and 25 once and halt the computation. The resulting label string becomes $a^{8n+4}b^{8(n+1)+8+4} = a^{8n+4}b^{8n+20}$. The $swastik$ pattern of each arm length $(n+3)$ is obtained in the output membrane 5.
\end{example}

\section{Main Results}

In the labelled array P system that we defined in section 3, the label language is a language over the labels such that for every string in the label language, there corresponds a halting computation of the P system that halts in one of the predefined final configurations. We also mentioned that the label can be $\lambda$. We now compare the label languages thus obtained with Chomsky hierarchy. \\

\noindent\textbf{Notation:} For any family $\mathcal{P}$ of languages, $\mathcal{P} \setminus \{ \lambda \}$ means the family of $\lambda$-free languages.

\begin{theorem} $(REG \setminus \{\lambda \}) \subseteq L8dAP_{1}$
\label{th2}
\end{theorem}
\begin{proof} Let $G = (N, T, P, S)$ be a $\lambda$-free regular grammar and let it generate the language $L$. We assume that each production in $P$ is of the form $A \rightarrow aB$, $A \neq B$ or $A \rightarrow a$, where $A$, $B \in N$, $a \in T$. Suppose the production is of the form $A \rightarrow aA$, to eliminate recursion, we replace $A \rightarrow aA$ with the set of new productions $A \rightarrow aA^{'}$ and $A^{'} \rightarrow aA$ by introducing a new non-terminal $A^{'}$ and also add productions $A^{'} \rightarrow cB$ for each production $A \rightarrow cB \in P$. So, we now obtain a new grammar $G^{'} = (N^{'}, T, P^{'}, S)$ from $G$. It is obvious that both the grammars $G$ and $G^{'}$ generate the same language (i.e.,) $L(G)= L(G^{'})$. Let $m$ be the number of non-terminals in $G^{'}$. Now, rename the non-terminals in $G^{'}$ as $A_i$ , $1\leq i \leq m$, such that $A_1 = S$ and also modify the productions with this renamed non-terminals. Now, we construct a $L8dAPS$ $\Pi_4$ with one membrane such that $L(G^{'}) = L8dAPS(\Pi_4)$ as follows: \\
\begin{center}
$\Pi_4$ = $\Big( N \cup \big{\{}\star\big{\}}, \big{\{}\star\big{\}},[_1]_1,\big{\{}A_1\big{\}}, R_1, 1, T \Big)$ 
\end{center}
\noindent where  \\

$R_1 = \Big{\{} a: A_i \rightarrow
(\star~A_j)^{\frac{\pi}{4}}~\big{|}~ A_i \rightarrow aA_j \in P^{'}\Big{\}}~\cup \Big{\{} a: A_i \rightarrow
(\star)^{\frac{\pi}{4}}~\big{|}~ A_i \rightarrow a \in P^{'}\Big{\}}$,  $F = \Big{\{}(\mathcal{D}) \big{|} \mathcal{D} \in L_{\star} \Big{\}}$ \\

\noindent The $L8dAPS$ $\Pi_4$ constructed above works as follows: Initially, the system contains a single symbol $A_1$, the start symbol of $G^{'}$. When the system chooses the rule $A_i \rightarrow (\star~A_j)^{\frac{\pi}{4}}$ which corresponds to the rule $A_i \rightarrow aA_j \in P^{'}$ and generate $\star$ in the direction of 45$^{\circ}$ and a rule with label $a$ is considered in the label string. Repeated application of such rules generates $\star$ in the direction of 45$^{\circ}$. When the system chooses the rule $A_i \rightarrow
(\star)^{\frac{\pi}{4}}$ which corresponds to the rule $A_i \rightarrow a \in P^{'}$ the computation halts. The system then generates the last $\star$ in the output array. The output array corresponding to the strings of length $1, 2, 3, \dots$ generated by $L8dAPS$ $\Pi_4$ are as in $L_{\star}$.
\begin{center}

$L_{\star} = \Bigg\{\begin{array}{cccccccccccccc}  
             &       &        &          &         &        &        &       &           
             &       &        &          &         & \star~   \\ 
                                  
             &       &        &          &         &        &        &\star~      &        
             &       &        &          & \star &          \\
             
             &       &        & \star~   &         &        & \star~ &       &         
             &       & \star~ &          &         &    \\
            
     \star~  &  ,~   & \star~ &          &  ,~     & \star~ &        &        &    
      ,~     &\star~ &        &          &         & ,~ \dots
\end{array}\Bigg\}$
\end{center}

\end{proof} 

\begin{remark} One can show that the context-free language $L = \Big\{ a{^n}b^{n} : n\geq 1 \Big\}$ can be a label language of $8dAPS$. Consider the following $L8dAPS$ $\Pi_5$, \\
\begin{center}
$\Pi_5 = \Big(\Big{\{}A, B, 0, \star \Big{\}}, \Big{\{} 0, \star\Big{ \}}, [_1[_2]_2]_1, I_1, I_2, R_1, R_2, 2,~\Big{\{} a, b \Big{\}}, F\Big)$
\end{center}
\noindent where ~~$R_1 =\Bigg{\{}1)~a: \begin{array}{c}
A
\end{array}\rightarrow \begin{array}{c}
(BA)^{0}
\end{array},~2)~a: \begin{array}{c}
A
\end{array}\rightarrow \begin{array}{c}
(B)^{0}
\end{array},in_2\Bigg{\}}$,\\

~$~~~~~R_2 =\Bigg{\{}3)~b: \begin{array}{c}
B
\end{array}\rightarrow \begin{array}{c}
(\star)^{0}
\end{array}\Bigg{\}}$, \\

~~~~~~~$I_1 = \Big{\{}A \Big{ \}}$, $I_2 = \phi$ and $F = \Big{\{}(\phi$, ${\star}^n)~ \big{|}~ n \geq 1  \Big{\}}$.

\end{remark}

\noindent Hence we can deduce the following:

\begin{proposition} $L8dAP_{\star}\setminus REG \neq \emptyset$
\label{pro1}
\end{proposition}

\begin{remark} We conclude from Proposition \ref{pro1} and Theorem \ref{th1} that $(REG\setminus \{ \lambda \})  \subset L8dAP_{\star}$. We proceed further to see whether $(CF \setminus \{\lambda \})  \subseteq L8dAP_{\star}$ which we prove in the following theorem.
\end{remark}

\begin{theorem} $(CF \setminus \{ \lambda \})\subseteq L8dAP_{1}$ \end{theorem}
\begin{proof} Let $L$ be a context-free language. Then let $G = (N, T, P, S)$ be a context-free grammar in $Greibach$ normal form generating $L$.  Let $n$ be the number of non-terminals in $N$. Now rename the non-terminals in $N$ as $A_i$, $1 \leq i \leq n$, such that $A_1 = S$ and also modify the rules with this renamed non-terminals. Let $G_1 = (N^{'}, T, P^{'}, A_1)$ be the grammar thus modified. We construct a $L8dAPS$ $\Pi_6$  with one membrane such that $L(G_1) = L8dAPS (\Pi_6)$  as follows: \\
$\Pi_6 = \Big( N^{'}\cup\{\star\}, \{ \star \}, [_1 ]_1, I_1, R_1, 1, T , F \Big)$, where $I_1 = \{A_1\}$ and \\
$R_1= \Bigg{\{}\begin{array}{c}
a:A_i \rightarrow \begin{array}{c}
(\star y)^{0}
\end{array}
\end{array}
:{A_i \rightarrow ay }\in P{'}\Bigg{\}} \cup \Bigg{\{}\begin{array}{c}
a:A_i \rightarrow \begin{array}{c}
(\star)^{0}
\end{array}
\end{array}
: {A_i \rightarrow a }\in P{'}\Bigg{\}}$

\noindent The arrays generated by $L8dAPS$ $\Pi_6$ are in $F  = \Big{\{}( \star \star \dots \star \star) : $ number of $\star's = \left| w \right|, w \in L8dAP(\Pi_6) \Big{\}}$.\\

\noindent Initially, the $L8dAPS$ $\Pi_6$ starts with an axiom $A_1$. Now, apply either a rule
$\begin{array}{c}
a:A_1 \rightarrow \begin{array}{c}
(\star y)^{0}
\end{array}
\end{array}$ which corresponds to $A_1 \rightarrow ay$ or the rule
$\begin{array}{c}
a:A_1 \rightarrow \begin{array}{c}
(\star)^{0}
\end{array}
\end{array}$ that corresponds to $A_1 \rightarrow a$. If we apply the latter rule, then the system halts on one of the final configurations in $F$. The corresponding label string generated by $\Pi_6$ is $a$. Suppose, we choose the former rule, $A_1$ is replaced with $\star y$ in the $0$ degree direction, $y$ is a string of non-terminals. We adopt the same procedure to the leftmost non-terminal in the array. Once we choose the rule $a :A_i \rightarrow (\star)^{0}$, then the leftmost non-terminal in the array is replaced by $\star$. The leftmost symbol in the array is not a non-terminal, so the non-terminal next to $\star$ is preferred. Again for this non-terminal we have two possibilities, either we can apply the rule $\begin{array}{c}
a:A_i \rightarrow \begin{array}{c}
(\star y)^{0}
\end{array}
\end{array}$ or the rule
$\begin{array}{c}
a:A_i \rightarrow \begin{array}{c}
(\star)^{0}
\end{array}
\end{array}$. Proceeding in this way, and applying the only possible rule $A_i \rightarrow (\star)^{0}$ to rewrite remaining non-terminals in the array as $\star$'s, the system halts on the final configuration $\Big{\{} \star \star \dots \star \star : $ number of $\star's = \left| w \right|, w \in L8dAP(\Pi_6) \Big{\}}$. Note that what ever $w$ may be, the halting array is of the form $\overbrace{\star \star\dots \star\star}^n$ with $\left| w \right| = n$. The label string is obtained by consuming a label of the rule in each step.
\end{proof}

\begin{remark} The context-sensitive language $\{ a{^n}b^{n}c^{n} : n\geq 1 \}$ can be generated by a $L8dAPS$ which gives the following proposition.
\end{remark}

\begin{proposition} $L8dAP_{\star}\setminus CF \neq \emptyset$ \end{proposition}


\begin{theorem} $CS \setminus L8dAP_{\star} \neq \emptyset$ 
\end{theorem}

\begin{proof} For the proof of this theorem we give a context-sensitive language which can not be a label language  of any $L8dAPS_{\star}$. Consider the context-sensitive language $L = \{ {a^{2}}^{n} ~\big{|}~  n \geq 0 \}$. Since $L$ is over one letter alphabet , all the rules in the $L8dAPS$ must be an $a$-rule and we cannot use  $\lambda$-label to any rule.  Let $(\alpha)^{\theta} \rightarrow \beta_k^{\theta}$ be a $\theta$-rotation rule such that $\alpha$ contains exactly one non-terminal (with zero or more number of terminals) and $\beta_k$ contains exactly $k$ non-terminals (with zero or more number of terminals). Suppose on the contrary, let, $L8dAPS$ $\Pi_7$ be a system with $m$ membranes that generates $L$ and halts on one of the final configurations. We show the non-existence of such a system only for $m = 1$. The argument for $m$ membrane $P$ system will be identical to this. The reason is that in both situations we need infinite number of rules in the membrane system to build $L$. Let $\mathcal{A}_1, \mathcal{A}_2,\dots,\mathcal{A}_n$ be the arrays in the initial configuration of $\Pi_7$. We recall that the successful halting computation means that the system must halt as well as the arrays remaining in the output membrane are terminal arrays (composed of only terminals).
\par In the following argument we actually look for rules in the membrane to build $L$ recursively.
\begin{enumerate}
\item In order to generate a label string $a$, whose length is one, the system must go up to one step (transition). Therefore, each array $\mathcal{A}_1, \mathcal{A}_2,\dots,\mathcal{A}_n$ in $\Pi_8$ must contain at most one non-terminal (no restriction on terminals). To reach the successful halting computation, we must apply one or more rules of type $a: \alpha^{\theta} \rightarrow \beta_0^{\theta}$. Note that we have introduced at least one new rule to generate the control string $a$.

\item By (1) above we know that each array $\mathcal{A}_1, \mathcal{A}_2,\dots,\mathcal{A}_n$ contains at most one non-terminal. To accept label string $a^2$, the system must go up to two steps (transitions). In order to do this, at least to one of the array, we need to apply the rule of type $a: \alpha^{\theta} \rightarrow \beta_1^{\theta}$, which is a new rule. This rule may be recursive or non-recursive . If the rule is recursive, then it also generates the label strings $a^3, a^5, a^6,$\dots$ \notin L$. Suppose, it is non-recursive, we can apply it once, followed by an existing rule of type $\alpha^{\theta} \rightarrow \beta_0^{\theta}$ to halt the computation. Hence, to  generate $a^2$, we have introduced a new rule of type $\alpha^{\theta} \rightarrow \beta_1^{\theta}$.

\item Similarly, in order to generate $a^4$, the system must go up to 4 steps. At least for one of the arrays we need to apply the rules in a way that there is no recursion. In all the possible cases, if any of the rule is recursive, it leads to the generation of a label string not in $L$.
Therefore, the only possibility is non-recursive rules. In all the cases, we can see at least one new rule is required to generate $a^4$.
\end{enumerate}

\par So, to generate each string in $L = \{ {a^{2}}^{n} ~\big{|}~ n \geq 0 \}$, we need to introduce at least one $a$-rule in each step. Since $L$ is infinite, the number of $a$-rules required to generate $L$ is also infinite.

\par Now we give the argumentative proof similar to the above to show that there does not exist any $L8dAPS$ to generate $L$. Suppose we assume that there is one such $L8dAPS$ $\Pi_7$ with $m$ membranes. If in any one of the $m$ membrane contain a recursive $a$-rule ,then it leads to an infinite loop or the system generates a string not in $L$. On the other hand, if the system contains only non recursive $a$-rules then, the number of such $a$-rules must be infinite. Hence the theorem.
\end{proof} 

\begin{theorem} $L8dAP_{\star} \subset CS$ \end{theorem}
\begin{proof} We show how $L8dAPS$ will be recognized by a linear
bounded automaton. In order to do this, we simulate a computation
of a $L8dAPS$ by remembering the number of symbols in the arrays and their corresponding shapes after the generating each symbol in the label string. We then show that the total number of symbols in the arrays is bounded by the length of the label string.

 Consider a label language $L$ of a $L8dAPS$ $\Pi_8$ with $m$ membranes and let $p$ be the total number of rules in these $m$ membranes. Let $w = b_1b_2\dots b_l$, $l \geq 1$ be a label string in $L$. Let $\mathcal{A}_1$, $\mathcal{A}_2,\dots ,\mathcal{A}_n$ be the arrays in the $m$ membranes of $\Pi_8$ in the initial configuration. We build a multi-track non-deterministic $LBA$ $B$ which simulates $\Pi_8$. In order for $B$ to simulate $\Pi_8$, it has to keep track the symbols in the arrays and their shapes after generating each symbol in the label string. So $B$ has a track assigned to every rule of $\Pi_8$, a track for each pixel-symbol triple
$(X, (x,y), i)$ $\in  V \times \mathbb{Z}^{2} \times \{1, 2, \dots, n \}$ and a track for each triple $(X, \mathcal{A}_i, j)\in V \times\{\mathcal{A}_1, \mathcal{A}_2,\dots,\mathcal{A}_n\} \times\{0, 1, 2,\dots \}$.

\par The array $\mathcal{A}_i$ is plotted in the plane $Z=i$ as follows: one of the symbols of the array $\mathcal{A}_i$ is plotted at $((0, 0), i)$, the origin of the plane $Z=i$. Fix this symbol, and place the other symbols of the array as follows: a symbol left to it is plotted at $((-1, 0), i)$; a symbol right to it is plotted at $((1, 0), i)$; a symbol above to it is plotted at $((0, 1), i)$; a symbol below to it is plotted at $((0, -1), i)$; a symbol 45 degree angle to it, is plotted at $((1, 1), i)$; a symbol 135 degree angle to it, is plotted at $((-1, 1), i)$; a symbol 225 degree angle to it, is plotted at $((-1, -1), i)$; a symbol 315 degree angle to it, is plotted at $((1, -1), i)$. In general, if the fixed symbol is in the position $((x,y), i)$, then a symbol to its left, is plotted at $((x-1,y), i)$; a symbol right to it, is plotted at $((x+1, y), i)$; a symbol above to it, is plotted at $((x,y+1), i)$; a symbol below to it, is plotted at $((x,y-1), i)$; a symbol 45 degree angle to it, is plotted at $((x+1, y+1), i)$; a symbol 135 degree angle to it, is plotted at $((x-1, y+1), i)$; a symbol 225 degree angle to it, is plotted at $((x-1, y-1), i)$; a symbol 315 degree angle to it, is plotted at $((x+1, y-1), i)$. If any symbol of the array remains, then change the fixed position of the symbol and repeat the same procedure till all the symbols in the array are plotted.

\par $B$ keeps track of the configuration of $\Pi_8$ by writing a positive integer 1 on each track assigned to the symbol-pixel triple $(X, (x,y), i)$, the symbol $X$ being plotted in the pixel $(x, y)$ of the plane $Z=i$. And also writes a positive integer on each track assigned to the symbol-configuration triple $(X, \mathcal{A}_i, j)$, denoting the number of symbols $X$ in the array $\mathcal{A}_i$ at the configuration $j$. Then for each triple $(X, (x,y), i)$, $B$ examines the chosen rule set and plots the symbols $X$ in the pixel $(x, y)$ of the plane $Z=i$ by the procedure mentioned above, increasing the number on the track $(X, \mathcal{A}_i, j)$ .

\par We can see that in any step of the computation, the tracks contain integers bounded by the number of symbols in a label string of $\Pi_8$. The shapes of the arrays are also preserved.

\par The number of symbols in the arrays in any configuration $C$
during a computation step is bounded by $S(i)$, where $i$ is the
number of label symbols generated. Then the space used by $B$ to record
the configurations and to calculate the configuration change of
$\Pi_8$ is bounded by $t \times log_b (S(i))$, where $b$ denotes the
base of the track alphabet and $t$ denotes the number of tracks
used. Finally, $B$ checks the applicability of some more rules.
If not, and also if it reaches one of the final configurations, it generates the label string $w$, otherwise it rejects. So the number of symbols in the arrays present in the system is
bounded by the input length and hence the label language is
a context-sensitive language.
\end{proof} 

\begin{theorem} $L_{\lambda}8dAP_{\star} = RE$ \end{theorem}

\begin{proof} The inclusion $L_{\lambda}8dAP_{\star} \subseteq RE$ follows from Church-Turing hypothesis.

\par For the proof of the inclusion $RE \subseteq L_{\lambda}8dAP_{\star}$, it is enough to prove that \\
$RE \subseteq L_{\lambda}8dAP_1$, since $L_{\lambda}8dAP_1 \subseteq L_{\lambda}8dAP_{\star}$.

\par Let $H = \{a_1, a_2,\dots,a_k\}$ and let $L \subseteq H^{*}$ be a recursively enumerable language. Let $e:
H \mapsto \{1^{1}, 1^{2}, 1^{3},\dots,1^{k}\}$ such that
$e(a_i) = 1^{i}$, $1 \leq i \leq k$. The encoding for a string $w
= a_ia_j\dots a_l$, for $a_i, a_j,\dots,a_l \in H$ is as follows:
$e(w) = 0e(a_i)0e(a_j)0\dots 0e(a_l)0$ .

\par For any $L$, there exists a Turing machine $M =(K, \{0,1\}, \Gamma,
\delta, q_0, F^{'})$ which halts after processing the input $i_0$
placed in its input tape if and only if $i_0 = e(w)$ for some $w
\in$ L. So it is sufficient to show the simulation of the encoding
$e(w)$, and the transitions of the Turing machine with a
$L8dAPS$ on $w$.

\begin{itemize}
  \item [$\bullet$] The transition $\delta(q, a) = (p, b, R)$ is simulated
  by the $\theta$-rotation rule \\
$\begin{array}{c}
(q~a~c)^0
\end{array}  \rightarrow \begin{array}{c}
(b~p~c)^0
\end{array}$~~, where $c$ is some non-blank symbol. \\

\item [$\bullet$] The transition $\delta(q, a) = (p, b, L)$ is simulated
  by the $\theta$-rotation rule\\
$\begin{array}{c}
(c~q~a)^0
\end{array}  \rightarrow \begin{array}{c}
(p~c~b)^0
\end{array}$,  where $c$ is some non-blank symbol.
\end{itemize}
We construct a $L8dAPS$ 
\begin{center}
$\Pi^{'} = \Big{(}V, T, [_1]_1, I_1, R_1, 1, H \Big{)}$
\end{center}

\noindent where $V = \Big{\{}q_0, q_1,\dots, q_k, 0, 1, x \Big{\}}$, \\
$~~~~~~~~~~T = \Big{\{} 0, 1, x \Big{\}}$, \\
$~~~~~~~~~~I_1 = \Big{\{}\begin{array}{c} 
q_0~0~e(a_i)~0~e(a_j)~0~\dots~0~e(a_l)~0
\end{array}\Big{\}}$, \\
$~~~~~~~~~~~R_1 = \Big{\{} a_i : \begin{array}{c}
(q_0~0~1)^0
\end{array}  \rightarrow \begin{array}{c}
(0~q_i~1)^0
\end{array}~:~ a_i \in H,~ {1 \leq i \leq k} \Big{\}} \cup$ the set of all $\theta$-rotation rules corresponding to the transitions of the Turing machine $M$ which are labelled with $\lambda$. 
The set of final configuration is  = $\Big{\{}\Big{(}e(w)x\Big{)}~:~w \in L \Big{\}}$. \\

\noindent The $L8dAPS^{a}$ $\Pi^{'}$  performs the following operations.
\begin{enumerate}
\item For $1 \leq i \leq k$, and the symbol $a_i \in H$, the rule
$\begin{array}{c}
(q_0~0~1)^0
\end{array}  \rightarrow \begin{array}{c}
(0~q_i~1)^0
\end{array}$, labelled with $a_i$ is used, which introduces the symbol $q_i$  and it is the symbol used in the first transition for generating the encoding in Step 2.
\item Perform the computation $e(au) = e(a)e(u)$, $u \in H^{+}$, $a \in H$. Assume that the encoding of $w$ is represented by encoding of each symbol of $u$ padded by $0$ on both ends. The simulation of $au$ is performed by the following sub-program. \\

    $\delta(q_i, 1) = (q_{i-1}, 1, R)$~,~~~~ $i = i,~ i-1,~ i-2,\dots,~ 3,~ 2$. \\

    $\delta(q_1, 1) = (q_0, 1, R)$ \\

\noindent The transitions of the sub-program can be simulated by the $\theta$-rotation rules as shown in the beginning of the proof, and these rules are assigned with label $\lambda$.
\item  Repeat the Steps 1 and 2 non-deterministically until the last symbol of the label string $w$ is consumed.
\item The output array that remains in the system is $\Big{\{}\Big{(}e(w)x\Big{)}~:~w \in L\Big{\}}$, which belongs to the set of final configuration. The array reduced in the system is equal to $e(w)x$ for some $w \in H^{+}$. We now start to simulate the working of the Turing machine $M$ for recognizing the string $e(w)$. If the Turing machine halts, by introducing the following transitions: \\

     $\delta(q_0, 0) = (0, x, R)$ \\

\noindent and if the corresponding $\theta$-rotation rule is labelled with $\lambda$, then $w \in L$. Otherwise the machine goes into an infinite loop.
    \end{enumerate}
So, we can see that the computation halts after generating a string $w$ if and only if $w \in L$.
\end{proof}

\section{Concluding Remarks}
In this paper we introduced a new array P system called labelled 8-directional array P system ($8dAPS$). The data structure `string' is interpreted as `turtle-like' graphs with a possibilities to turn in multiples of 45 degrees. This P system-based fractal description model can be used to construct several interesting pictures. In the labelled $8dAPS$ the rules are labelled and the evolution of the system yields a label language. The label language thus obtained is compared with Chomsky hierarchy. We can understand the halting nature of this P system by means of its `dependability'. By dependability we mean to study the halting nature or halting configurations of P system via string over the labels of the rules. We know that in our model every rule is labelled and strings over the set of labels lead the application of the rules. Such strings decide the strategy of movements in the parallel distributed computing model, P system. Hence the study becomes significant. The main difference between grammar rewriting system of describing  some space filling curves like `Koch curve' and our recursive 8-directional array P system is that, we do not re-scale the template. If we are able to record the shrinking effect, then our $8dAPS$ can generate almost all curves like `Koch curve', Peano curve etc. One can also extend the study to understand more about space filling curves which have important role in antenna designing.
One of the other future work can also looking at control languages of other variants of P-systems\cite{Gh.Paun-Rozenberg-Salomaa-2010} such as Tissue P-systems with arrays as data structures.
\section{Acknowledgements}
This work was funded by the SERB :SB/S4/MS-865/14, Department of Science and Technology, Government of India.

\end{document}